\documentclass[preprint,12pt]{elsarticle}

\usepackage{amsmath,amsthm,amssymb}
\usepackage{enumitem}
\newtheorem{theorem}{Theorem}
\newtheorem{prop}{Proposition}
\usepackage{graphicx}
\usepackage{subcaption}  

\journal{Neural Networks}

\begin{document}

\begin{frontmatter}



\title{Versatile Auxiliary Regressor with Generative Adversarial network (VAR+GAN)}


\author{Shabab Bazrafkan, Peter Corcoran}

\address{National University of Ireland Galway}

\begin{abstract}

Being able to generate constrained samples is one of the most appealing applications of the deep generators. Conditional generators are one of the successful implementations of such models wherein the created samples are constrained to a specific class. In this work, the application of these networks is extended to regression problems wherein the conditional generator is restrained to any continuous aspect of the data. A new loss function is presented for the regression network and also implementations for generating faces with any particular set of landmarks is provided. 

\end{abstract}


\end{frontmatter}


\section{Introduction}
\label{sec:intro}
Generative Adversarial Networks (GAN) \cite{GAN} are among the most successful implementations of deep generators. The idea of GAN is to train two agents, a generator, and a discriminator, simultaneously. The generator is a deep neural network which accepts a vector from a latent space (uniformly distributed noise) and outputs a sample, same type of the database. The discriminator is a binary classifier determining whether this sample is generated or is a genuine data coming from the database. The training is accomplished by playing a min-max game between these two networks. There are several extensions to the original GAN idea wherein the original GAN is adapted to a specific condition by changing the network structures and/or loss function. For example, Conditional GAN (CGAN) \cite{CGAN}, Auxiliary Classifier GAN (ACGAN) \cite{ACGAN}, and Versatile Auxiliary Classifier with GAN (VAC+GAN) \cite{VACGAN} are utilizing the original GAN idea to train conditional generators wherein the output of the generator is constrained to a specific class given the right input sequence. CGAN does this by partitioning the latent space and also the auxiliary knowledge of the data class. In ACGAN the loss of the CGAN is manipulated by adding a classification term which back-propagates through generator and discriminator. The VAC+GAN extends the ACGAN scheme to be more adaptable to different GAN variations. This is done by adding a classifier network in parallel with the discriminator network, and the classification error is back-propagated through the generator.\\
In this work, the idea of VAC+GAN is extended to regression problems by replacing the classifier with a regression network. A new loss function is presented for this network. The regression error is back-propagated through the generator. This gives the opportunity to train a generator while constraining the generated samples to any continuous aspect of the original database.\\
Similar ideas include the scheme presented in \cite{CBIGAN} wherein a Hierarchical Generative Model (HGM) is utilized for eye image synthesis and eye gaze estimation. This work introduces a variation of GAN known as conditional Bidirectional GAN (cBiGAN) which is a mixture of CGAN and Bidirectional GAN (BiGAN). The main issue with this method is the lack of adaptability to every GAN variation. This method is only applicable to BiGAN scheme. This approach is shown in figure \ref{fig:1}. In our observations, cBiGAN implementation is able to generate samples for each aspect but there are very low variations in generated samples for a specific aspect. The proposed VAR+GAN scheme produces higher variations in the same condition. This is discussed in more details in section \ref{sec:results}. The other advantage of the proposed method is its versatility, i.e., it applies to any GAN implementation regardless of the network architecture and/or loss function.\\
\begin{figure}
  \includegraphics[width=\columnwidth]{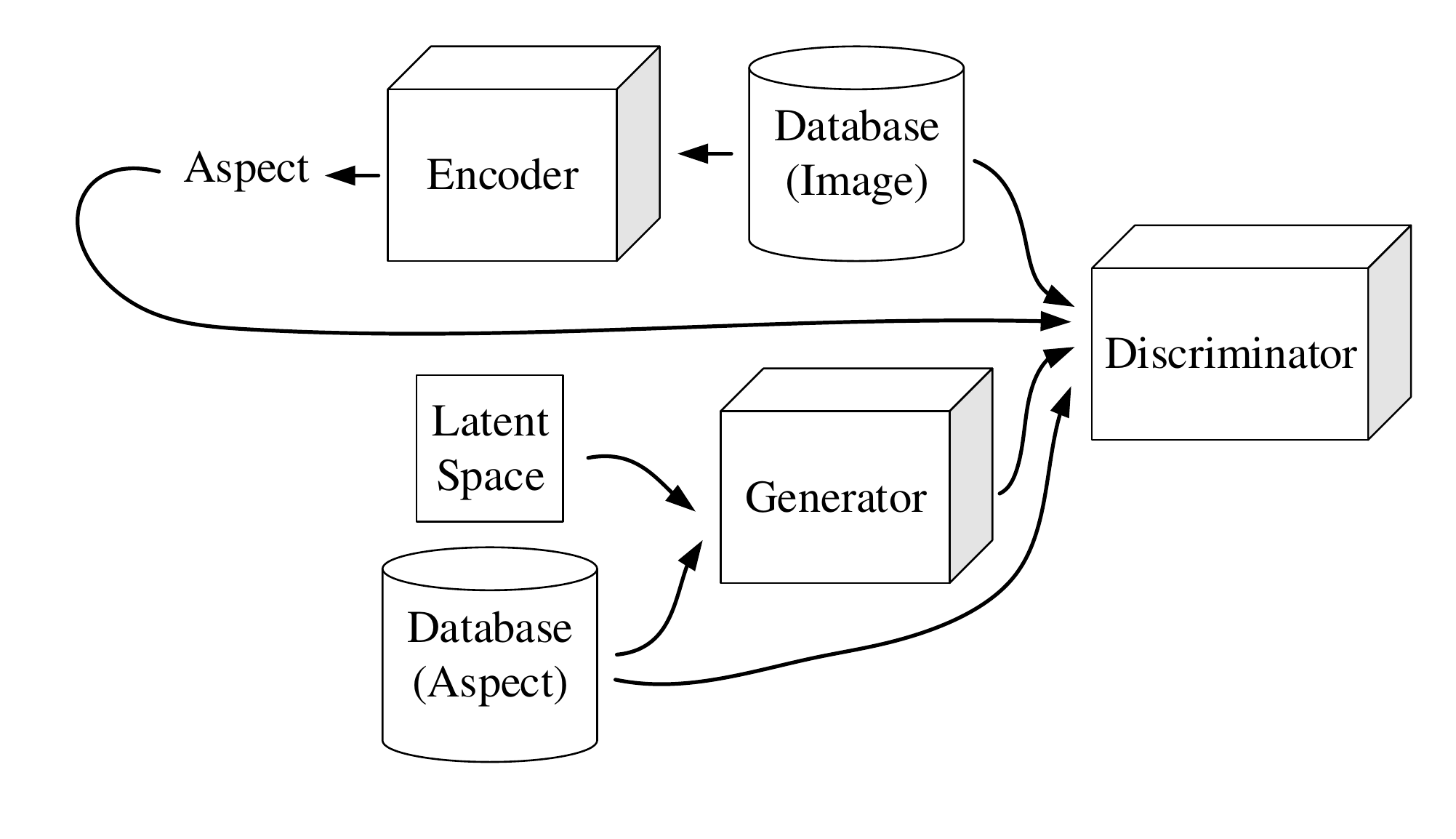}
\caption{cBiGAN scheme used for training conditional generators.}
\label{fig:1}       
\end{figure} 
In the next section, the idea of VAR+GAN is explained alongside with the presented loss function for regression network. Section \ref{sec:implementationAndResults} explains implementation, results and the comparisons for the presented method against cBiGAN and the conclusions and future works is discussed in the last section.

\section{Versatile Auxiliary Regressor with Generative Adversarial network (VAR+GAN)}
\label{sec:2}
The idea of proposed scheme is to place a regression network in parallel with the discriminator and back-propagate the regression error through the generator (see figure \ref{fig:2}). In this method, the generator is constrained to generate samples with specific continuous aspects. For example, in the face generation application, given the right latent sequence, the generator creates faces with particular landmarks.
\begin{figure}
  \includegraphics[width=\columnwidth]{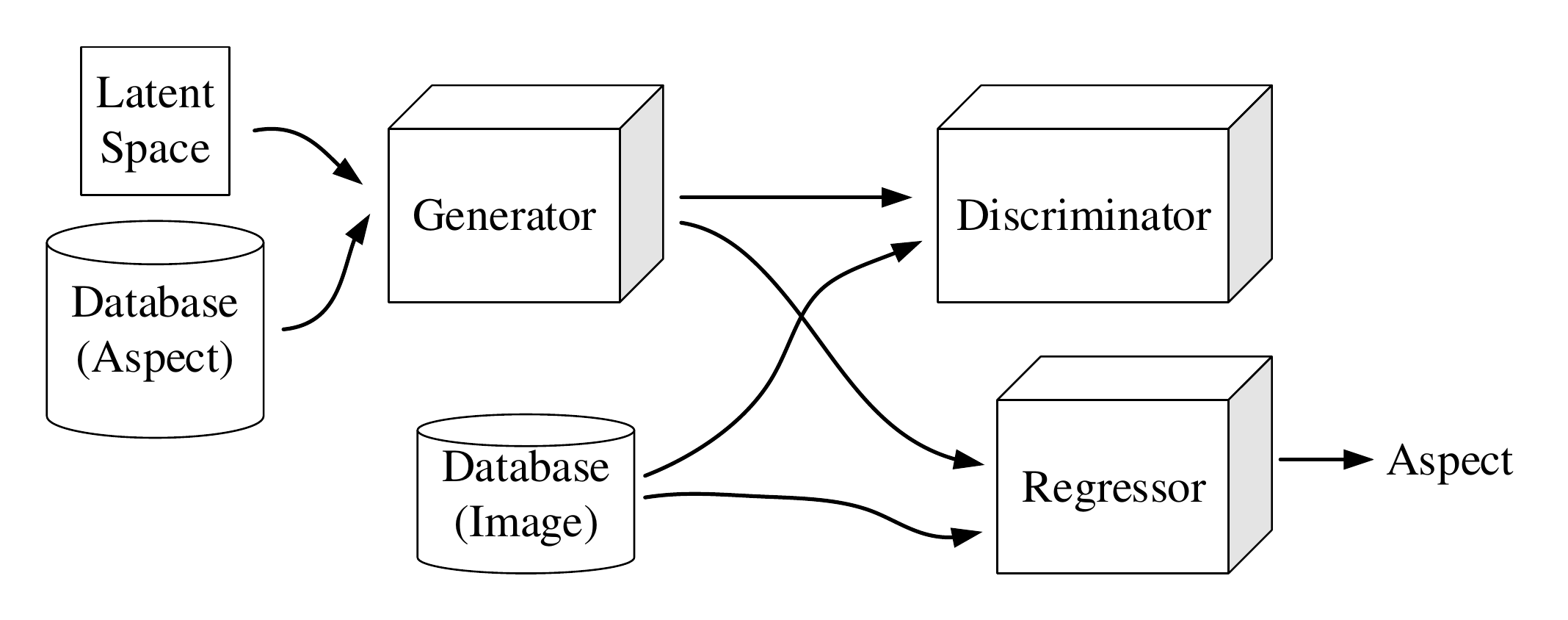}
\caption{Presented method (VAR+GAN) for training conditional generators.}
\label{fig:2}       
\end{figure}  The following loss function is introduced for the regression network
\begin{equation}
L_R = \int\int dp_z({\bf z})\bigg(-\log\Big(1-(y-R(G({\bf z})))\Big)\bigg)dz
\label{eq:LossR}
\end{equation}
wherein ${\bf z}$ is the latent space variable, $dp_z({\bf z})$ is the distribution of an infinitesimal partition of latent space, $y$ is the target variable (ground truth), $R$ is the regression function and $G$ is the generator function.
\begin{prop}
For the loss function in equation \ref{eq:LossR} the optimal regressor is
\begin{equation}
R^* = \frac{p({\bf x})}{c}+y-1
\label{eq:OptR}
\end{equation}
wherein $p({\bf x})$ is the distribution of the generator's output, $c$ is post-integration constant, and $y$ is the target function. 
\end{prop}
\begin{proof}
Considering the inner integration of equation \ref{eq:LossR} and by replacing $G({\bf z})={\bf x}$, the extremum of the loss function with respect to $R$ is
\begin{equation}
\frac{d}{dR}\int dp_x({\bf x})\bigg(-\log\Big(1-(y-R({\bf x}))\Big)\bigg)dx=0
\end{equation} 
which can be written as
\begin{equation}
\int\frac{-dp_x}{R-y+1}=0~~~\Rightarrow~~~\frac{p_x}{R-y+1}=c
\end{equation}
this results in
\begin{equation}
R = \frac{p({\bf x})}{c}+y-1
\end{equation}
concluding the proof.
\end{proof}
\begin{theorem}
Minimizing the loss function in equation \ref{eq:LossR} decreases the entropy of the generator's output.
\end{theorem}
\begin{proof}
by replacing equation \ref{eq:OptR} in \ref{eq:LossR} we have
\begin{equation}
L_R = \int\int-\log\Big(\frac{p_x({\bf x})}{c}\Big)dp_xdx
\end{equation}
which can be rewritten as
\begin{equation}
L_R = -\int p_x({\bf x})\log(p_x({\bf x}))dx + \log(c) = H(p_x({\bf x}))+\log(c)
\end{equation}
wherein $H$ is the shannon entropy. Minimizing $L_R$ results in decreasing $H(p_x({\bf x}))$ concluding the proof.
\end{proof}
Adding the regressor to the model decreases the entropy of the generated
samples. This is expectable since the idea is to constrain the output of the
generator to obey some particular criteria. This is shown in observations in
section \ref{sec:results}.
\begin{theorem}
For any two sets of samples and their corresponding targets ($y_1$ and $y_2$), the loss function in equation \ref{eq:LossR} increases the Jensen Shannon Divergence (JSD) between generated samples for these two sets.
\end{theorem}
\begin{proof}
Consider ${\bf z}_1$ and ${\bf z}_2$ are two partitions of the latent space correspond to two sets of samples with targets $y_1$ and $y_2$. In this case, the loss function in equation \ref{eq:LossR} is given by:
\begin{equation}
\begin{split}
L_R = &-\int p_{z_1}({\bf z})\log(1-(y_1-R(G({\bf z}_1))))dz\\
&-\int p_{z_2}({\bf z})\log(1-(y_2-R(G({\bf z}_2))))dz
\end{split}
\label{T1P1}
\end{equation}
Considering $G({\bf z}_1)={\bf x}_1$, $G({\bf z}_2)={\bf x}_2$, $c_1=1-y_1$, and $c_2=1-y_2$  equation \ref{T1P1} simplifies to
\begin{equation}
L_R = -\int p_{x_1}({\bf x}) \log(c_1+R({\bf x})) dx-\int p_{x_2}({\bf x}) \log(c_2+R({\bf x})) dx
\label{T1P2}
\end{equation}
To find the optimum $R(x)$ the derivative of the integrand is set to zero given by
\begin{equation}
\frac{p_{x_1}}{c_1+R}+\frac{p_{x_2}}{c_2+R}=0
\end{equation}
which results in
\begin{equation}
R = -\frac{p_{x_1}c_2+p_{x_2}c_1}{p_{x_1}+p_{x_2}}
\label{T1P4}
\end{equation}
By replacing equation \ref{T1P4} in equation \ref{T1P2} it simplifies to
\begin{equation}
L_R = -\int p_{x_1}\log\bigg(\frac{(c_1-c_2)p_{x_1}}{p_{x_1}+p_{x_2}}\bigg)+\int p_{x_2}\log\bigg(\frac{(c_2-c_1)p_{x_2}}{p_{x_1}+p_{x_2}}\bigg)dx
\end{equation}
which can be rewritten as
\begin{equation}
\begin{split}
R_L = &-\int \log \Bigg(\frac{p_{x_1}}{\frac{p_{x_1}+p_{x_2}}{2}}\Bigg)-\int \log \Bigg(\frac{p_{x_2}}{\frac{p_{x_1}+p_{x_2}}{2}}\Bigg)\\&-\log(c_1-c_2)-\log(c_2-c_1)-\log(4)
\end{split}
\end{equation}
which equals to
\begin{equation}
R_L=-\log(c_1-c_2)-\log(c_2-c_1)-\log(4)-2JSD(p_{x_1}||p_{x_2})
\end{equation}
minimizing $R_L$ increasing $JSD(p_{x_1}||p_{x_2})$ term, concluding the proof.
\end{proof}
In this section, it has been shown that the presented loss function increases the distance between generated samples for any two set of aspects, which is desirable.
\section{Implementation and Results}
\label{sec:implementationAndResults}
In this section, the implementation of the VAR+GAN is presented and compared to cBiGAN method. To keep the consistency in comparisons, the same architecture for generator network has been kept throughout all implementations. All the networks are trained using Lasagne \cite{LASAGNE} library on the top of Theano \cite{THEANO} library in Python.
\subsection{Network Architectures}
Three main architectures used in this section are encoder, decoder, and regression networks. The first two are shown in figure \ref{fig:architectures}. 
The decoder contains one fully connected layer which maps the input to a 3D layer. Next layers are all convolutional layers followed by $(2, 2)$ un-pooling layers for every second convolution. The exponential linear unit (ELU) \cite{ELU} is used as activation function except in the last layer wherein no non-linearity has been applied except for the encoder in cBiGAN scheme wherein tanh nonlinearity is applied in the output layer. The encoder network is made of convolutional layers with ELU activation function. The downscaling in these layers is obtained by using $(2, 2)$ stride in every second convolutional layer.In the decoder network, all convolutional layers have 64 channels while in the encoder, the number of the channels is gradually increased to 128, 192, and 256 after each pooling layer. The layers shown in red are applying no nonlinearity to their input.
\begin{figure}
\begin{subfigure}[b]{.49\textwidth}
  \includegraphics[width=\columnwidth]{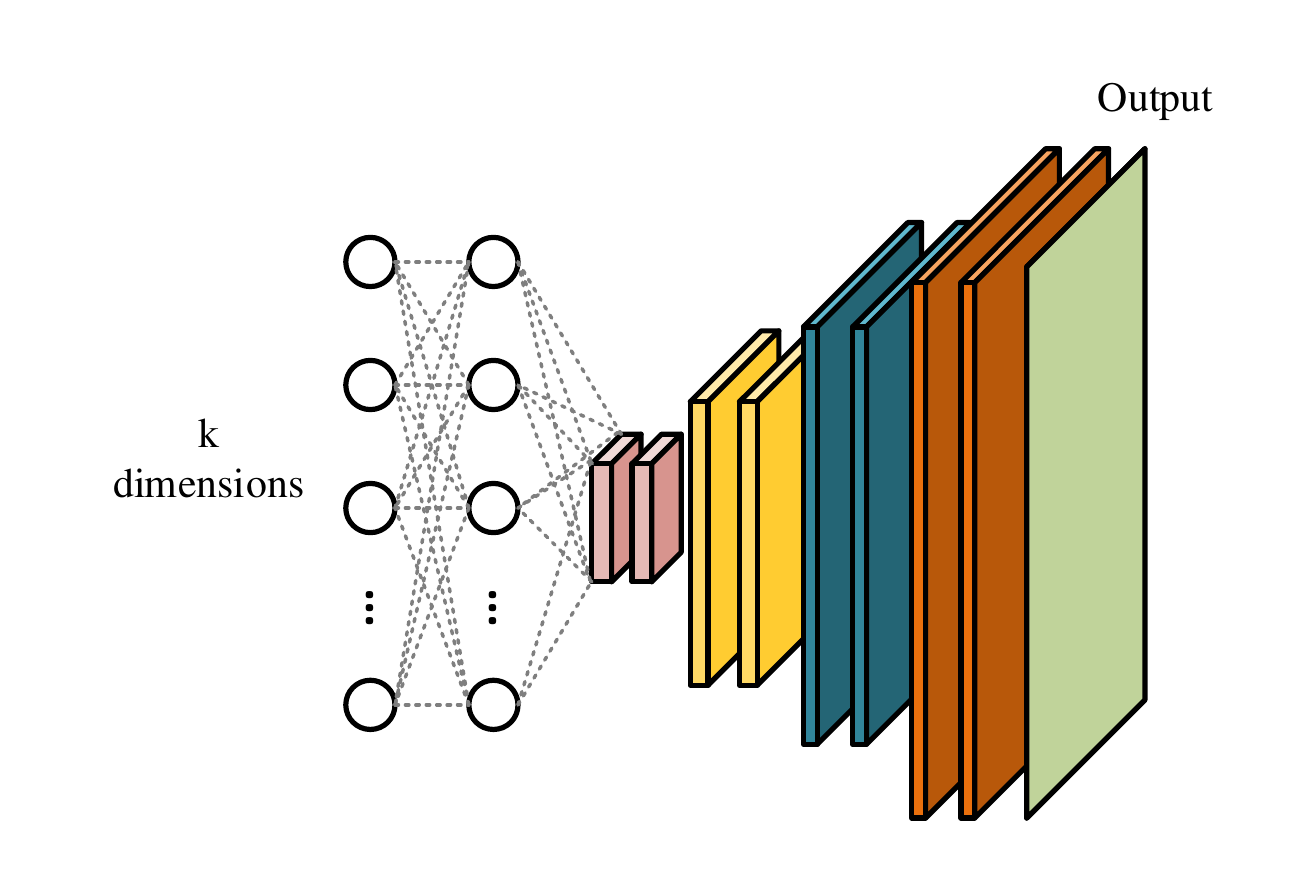}
\caption{Decoder network.}
\label{fig:decoder}       
\end{subfigure}
\begin{subfigure}[b]{.49\textwidth}
  \includegraphics[width=\columnwidth]{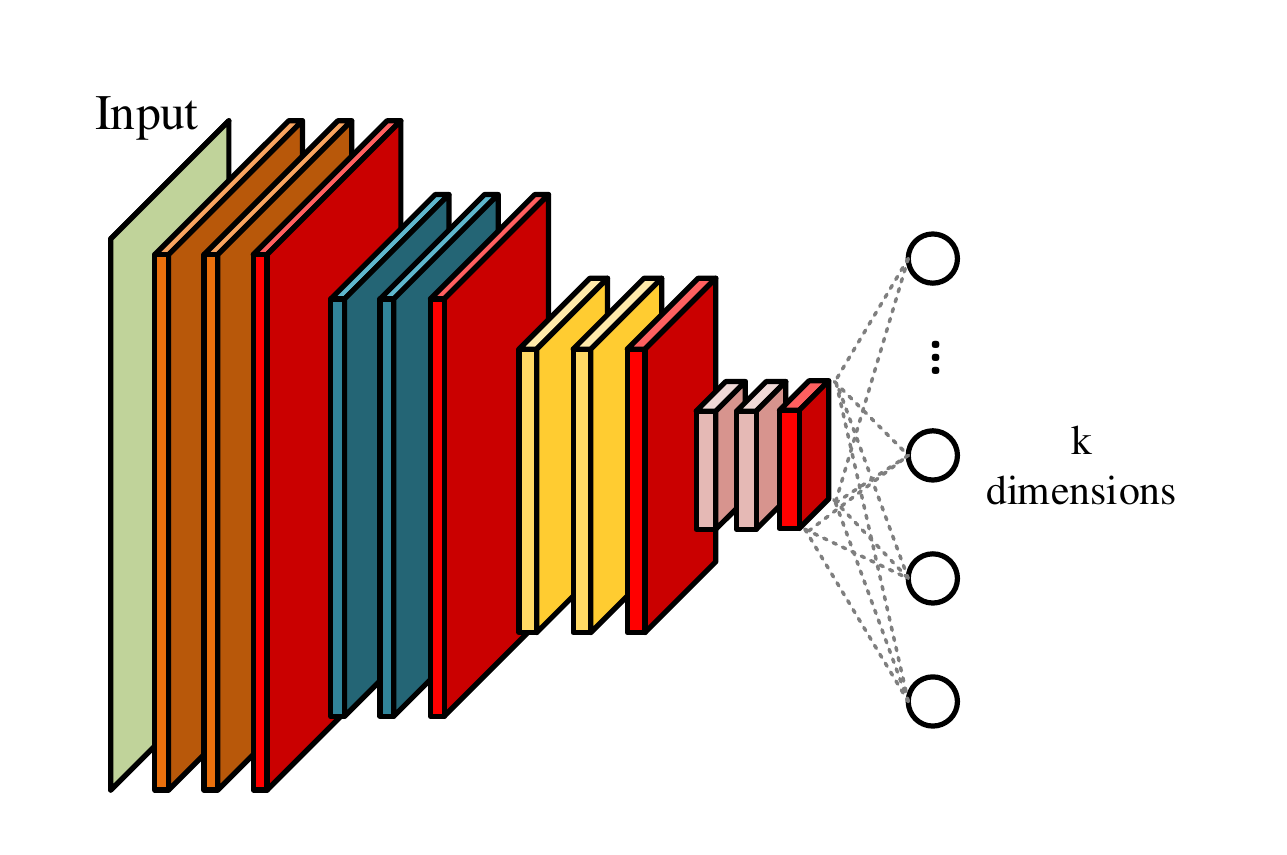}
\caption{Encoder network.}
\label{fig:encoder} 
\end{subfigure}
\caption{Network architectures used for implementation purposes.}
\label{fig:architectures}
\end{figure}
The regression network is a conventional deep neural network shown in table \ref{tab:1}.
\begin{table}
\caption{the Regression network used in experiments.}
\label{tab:1}
\centering       
\begin{tabular}{|l|l|l|l|}
\hline
Layer & Type & kernel & Activation  \\
\hline
Input & Input$(48\times 48)$ & -- & -- \\
\hline
Hidden 1 & Conv & $3\times3$(64 ch)& ReLU \\
\hline
Pool 1 &Max pooling &$2\times2$ &--\\
\hline
Hidden 2 & Conv & $3\times3$(64 ch)& ReLU \\
\hline
Pool 2 &Max pooling &$2\times2$ &--\\
\hline
Hidden 3 & Dense & 1024& ReLU \\
\hline
Output & Dense & 98 & Tanh \\
\hline
\end{tabular}
\end{table}
\subsection{Database}
The dataset used in this work is CelebA database \cite{CelebA} which is made of 202,599 frontal posed images. Face regions are cropped and resized to $48\times 48$ pixels using OpenCV frontal face cascade classifier \cite{H1}. Supervised Descent Method (SDM) \cite{H2} is used for facial point detection. The detector is based on \cite{H3} and it utilizes the discriminative 3D facial deformable model to find 49 facial landmarks including contours of eyebrows, eyes, mouth and the nose. These landmarks are used as the data aspect in this work.
\subsection{Implementation}
\subsubsection{VAR+GAN}
For the proposed scheme (figure \ref{fig:2}), the Boundary Equilibrium  Generative Adversarial Network (BEGAN) \cite{BEGAN} is utilized to train the deep generator. In this method the generator architecture is same as the decoder network shown in figure \ref{fig:decoder} with input dimension $k=128$. The discriminator is an auto-encoder network wherein the decoder architecture is same as the generator and the encoder is the network shown in figure \ref{fig:encoder} with $k=128$. The regression network is shown in table \ref{tab:1} and the loss function for proposed implementation is a modified version of the original BEGAN loss \cite{BEGAN} given by:
\begin{equation}
\begin{split}
&L_d = L(x)-k_t\cdot L(G(z|y))\\
&L_g = \vartheta \cdot L(G(z|y)) + \zeta\cdot L_R\\
&k_{t+1}=k_t+\lambda_{k}\big(\gamma L(x)-L(G(z|y))\big)
\end{split}
\end{equation}
Where $L_g$ and $L_d$ are generators and discriminators losses respectively. $G$ is the generator function, $z$ is a sample from the latent space, $x$ and $y$ are genuine image and corresponding ground truth drawn from the database , $\lambda_k$ is the learning rate for $k$, $\gamma$ is the equilibrium hyper parameter set to 0.5 in this work, $L_R$ is the regression loss given by equation \ref{eq:LossR}, and $\vartheta$ and $\zeta$ are set to 0.97 and 0.03 respectively, and $L$ is the auto-encoders loss defined by
\begin{equation}
L(v) = |v-D(v)|^2
\end{equation} 
The optimizer used for training the generator and discriminator is ADAM with learning rate, $\beta_1$ and $\beta_2$ equal to 0.0001, 0.5 and 0.999 respectively. And the regression network is optimized using nestrov momentum gradient descent with learning rate and momentum equal to 0.01 and 0.9 respectively.\\
\subsubsection{cBiGAN}
The cBiGAN scheme is implemented for the same generator architecture and on the same database to make fair comparisons with the proposed method. The generator architecture is same as the decoder network shown in figure \ref{fig:decoder} with input dimension $k=128$. The discriminator model is same as the encoder network in figure \ref{fig:encoder} with $k=1$ and $sigmoid$ non-linearity at the output layer. And the encoder network in figure \ref{fig:1} has the architecture shown in figure \ref{fig:encoder} with $k=98$ and $tanh$ non-linearity at the output layer. The loss function for this scheme is presented in \cite{CBIGAN} given by
\begin{equation}
\begin{split}
&L_D = \log(p^r)+\log(1-p^I)+\log(1-p^s)\\
&L_G = \log(p^I)\\
&L_E = \log(p^s)+\theta||s^--y||^2
\end{split}
\end{equation}
wherein
$s^-$ is the encoder's output, $y$ is the genuine aspect coming from the database, and 
\begin{equation}
p^r = D(y,x)~~~,~~~
p^I = D(y,G(z|y))~~~,~~~
p^s = D(s^-,x)
\end{equation}
where $D$ and $G$ are discriminator and generator functions respectively, $x$ and $y$ are genuine image and corresponding ground truth drawn from the database, and $z$ is a sample from the latent space. The coefficient $\theta$ is set to $0.8$. The optimizer used for training the model is ADAM with learning rate, $\beta_1$ and $\beta_2$ equal to 0.0001, 0.5 and 0.999 respectively.
\subsection{Results}
\label{sec:results}
In this section, the proposed method is compared against the cBiGAN method while generating faces for a particular landmark point set. The results for six sets of landmarks are shown in figures \ref{fig:results0} to \ref{fig:results77911}. In each figure (left) the outputs from the proposed method for a particular set of landmarks are illustrated while in right side of the figures the output of the generator trained in cBiGAN scheme is given for the same landmarks. 
\begin{figure}
\begin{subfigure}[b]{.49\textwidth}
  \includegraphics[width=\columnwidth]{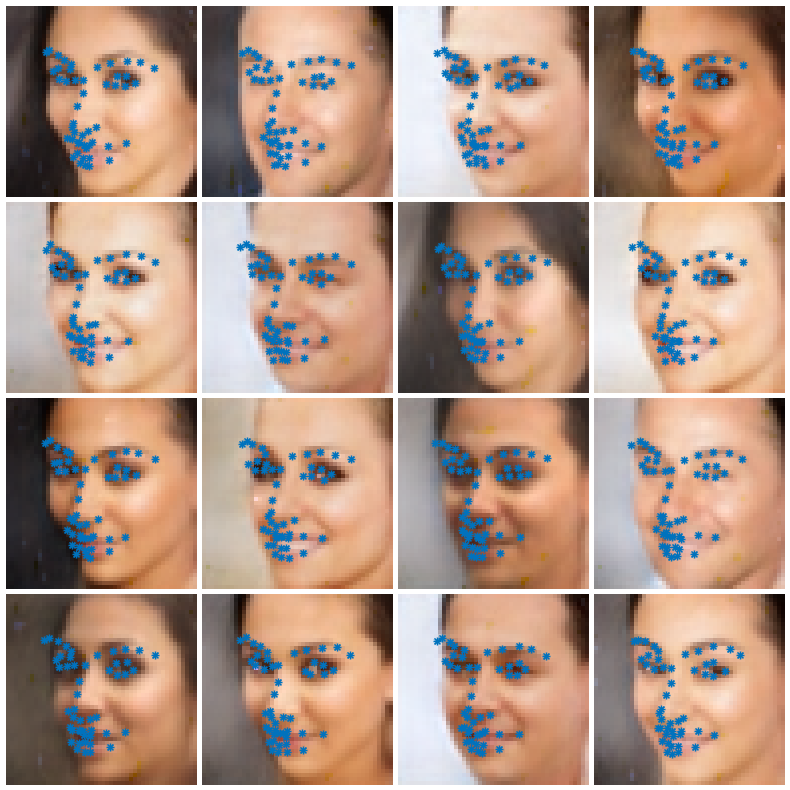}
\caption{Proposed method (VAR+GAN).}
\label{fig:VARGAN0}       
\end{subfigure}
\begin{subfigure}[b]{.49\textwidth}
  \includegraphics[width=\columnwidth]{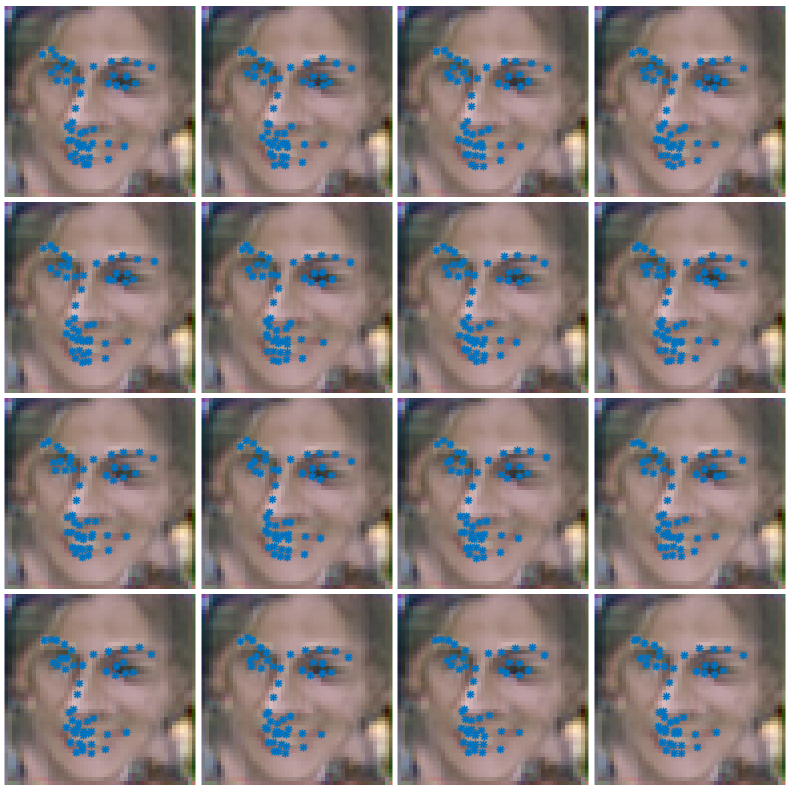}
\caption{cBiGAN.}
\label{fig:cBiGAN0} 
\end{subfigure}
\caption{Generator outputs for proposed method (VAR+GAN) vs cBiGAN.}
\label{fig:results0}
\end{figure}

\begin{figure}
\begin{subfigure}[b]{.49\textwidth}
  \includegraphics[width=\columnwidth]{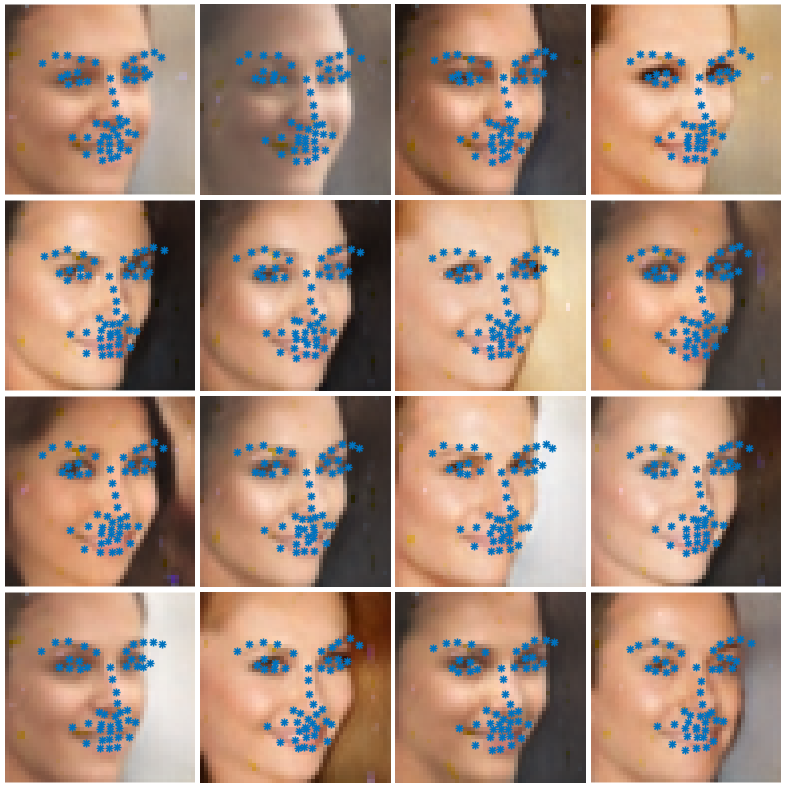}
\caption{Proposed method (VAR+GAN).}
\label{fig:VARGAN6}       
\end{subfigure}
\begin{subfigure}[b]{.49\textwidth}
  \includegraphics[width=\columnwidth]{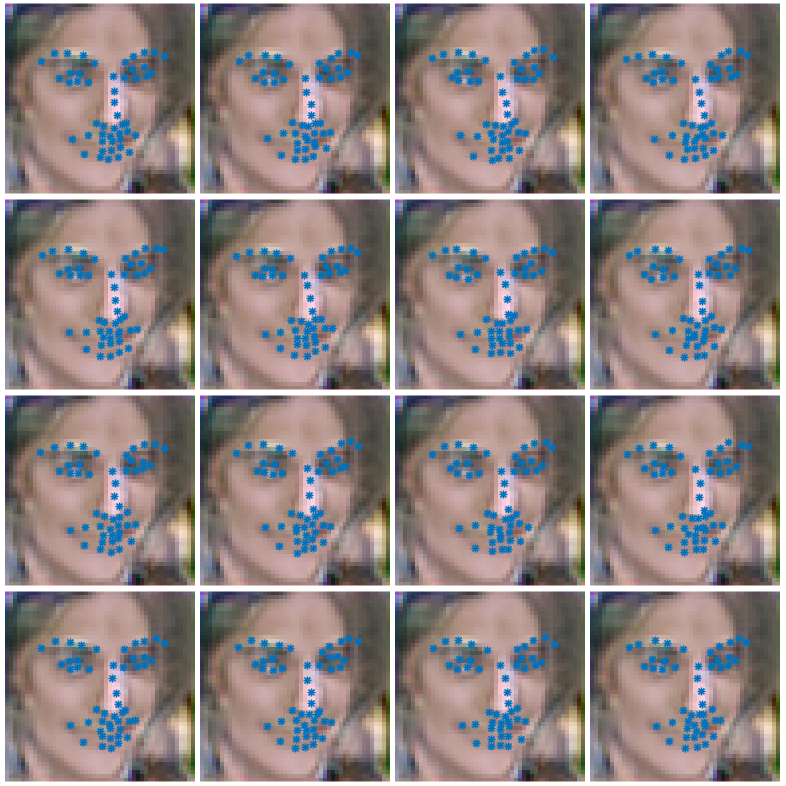}
\caption{cBiGAN.}
\label{fig:cBiGAN6} 
\end{subfigure}
\caption{Generator outputs for proposed method (VAR+GAN) vs cBiGAN.}
\label{fig:results6}
\end{figure}

\begin{figure}
\begin{subfigure}[b]{.49\textwidth}
  \includegraphics[width=\columnwidth]{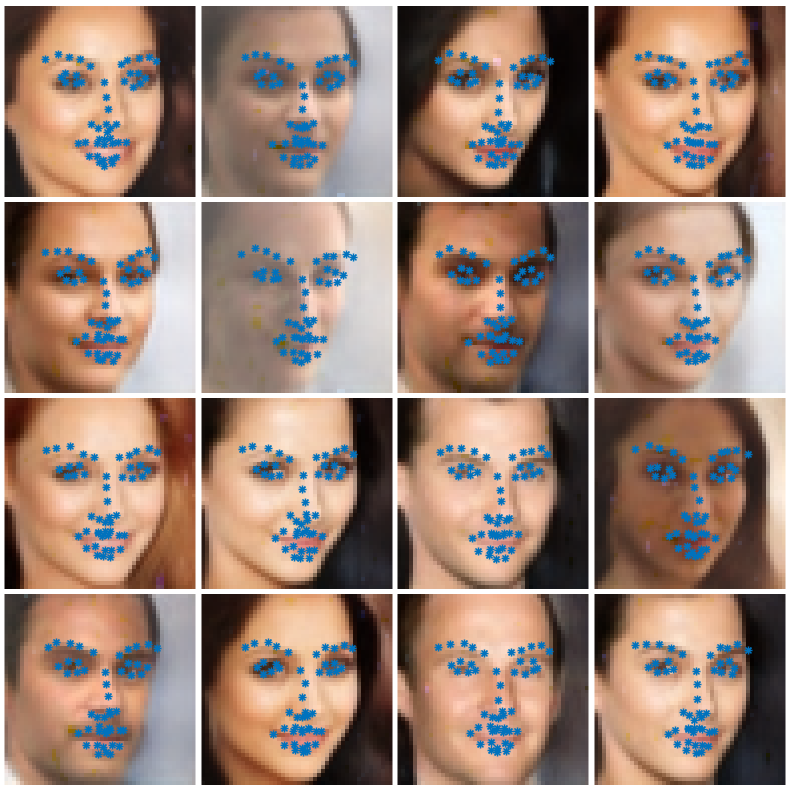}
\caption{Proposed method (VAR+GAN).}
\label{fig:VARGAN8}       
\end{subfigure}
\begin{subfigure}[b]{.49\textwidth}
  \includegraphics[width=\columnwidth]{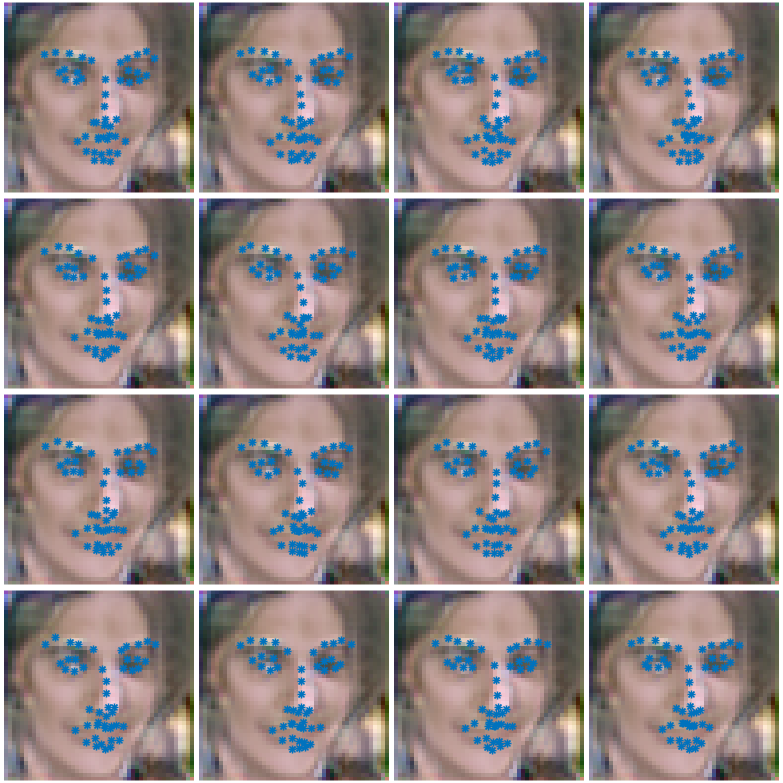}
\caption{cBiGAN.}
\label{fig:cBiGAN8} 
\end{subfigure}
\caption{Generator outputs for proposed method (VAR+GAN) vs cBiGAN given particular landmarks.}
\label{fig:results8}
\end{figure}

\begin{figure}
\begin{subfigure}[b]{.49\textwidth}
  \includegraphics[width=\columnwidth]{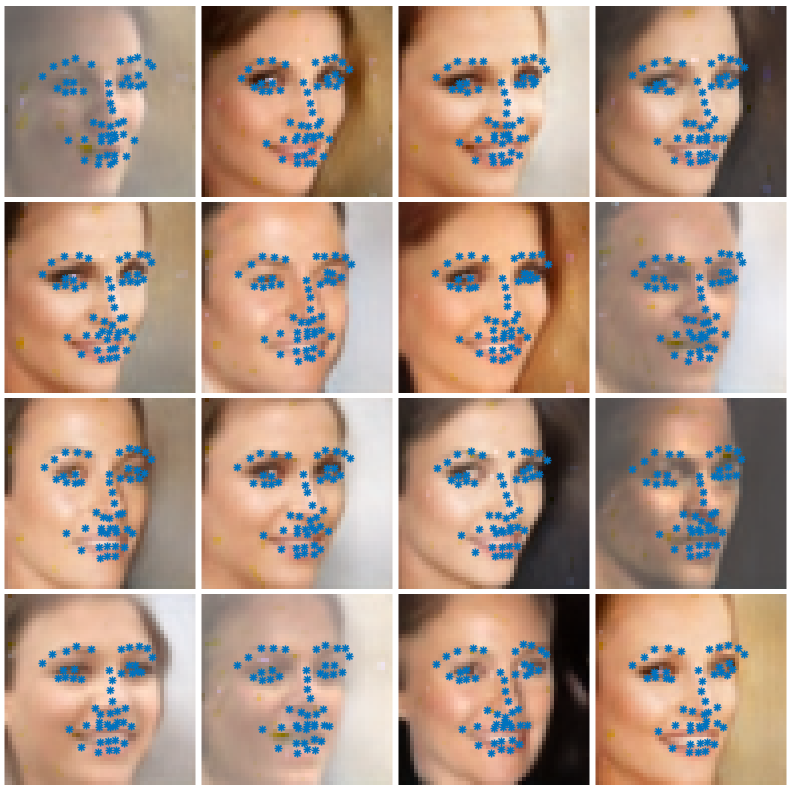}
\caption{Proposed method (VAR+GAN).}
\label{fig:VARGAN911}       
\end{subfigure}
\begin{subfigure}[b]{.49\textwidth}
  \includegraphics[width=\columnwidth]{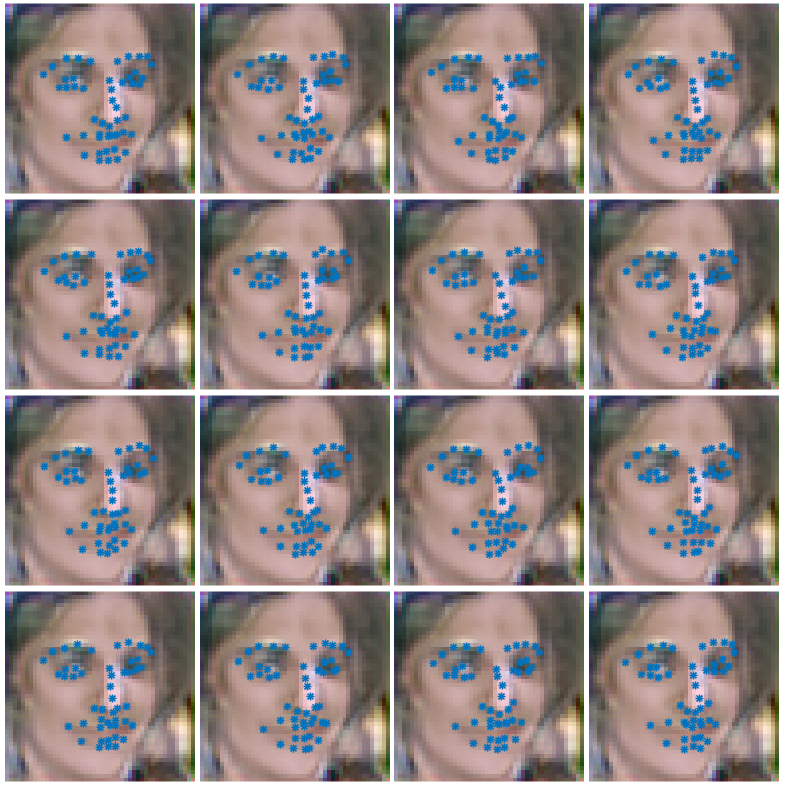}
\caption{cBiGAN.}
\label{fig:cBiGAN911} 
\end{subfigure}
\caption{Generator outputs for proposed method (VAR+GAN) vs cBiGAN given particular landmarks.}
\label{fig:results911}
\end{figure}

\begin{figure}
\begin{subfigure}[b]{.49\textwidth}
  \includegraphics[width=\columnwidth]{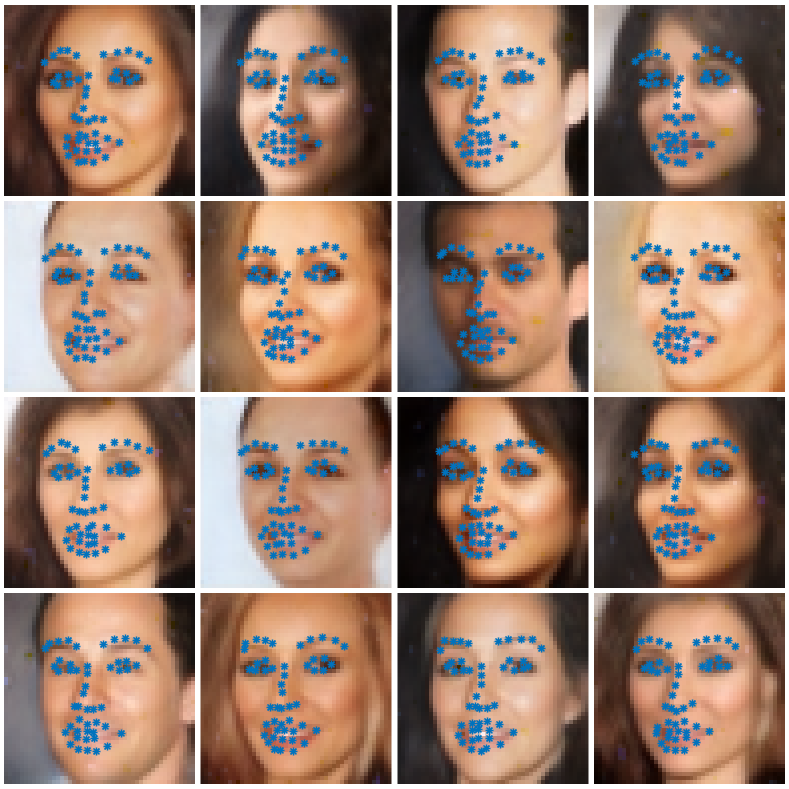}
\caption{Proposed method (VAR+GAN).}
\label{fig:VARGAN7911}       
\end{subfigure}
\begin{subfigure}[b]{.49\textwidth}
  \includegraphics[width=\columnwidth]{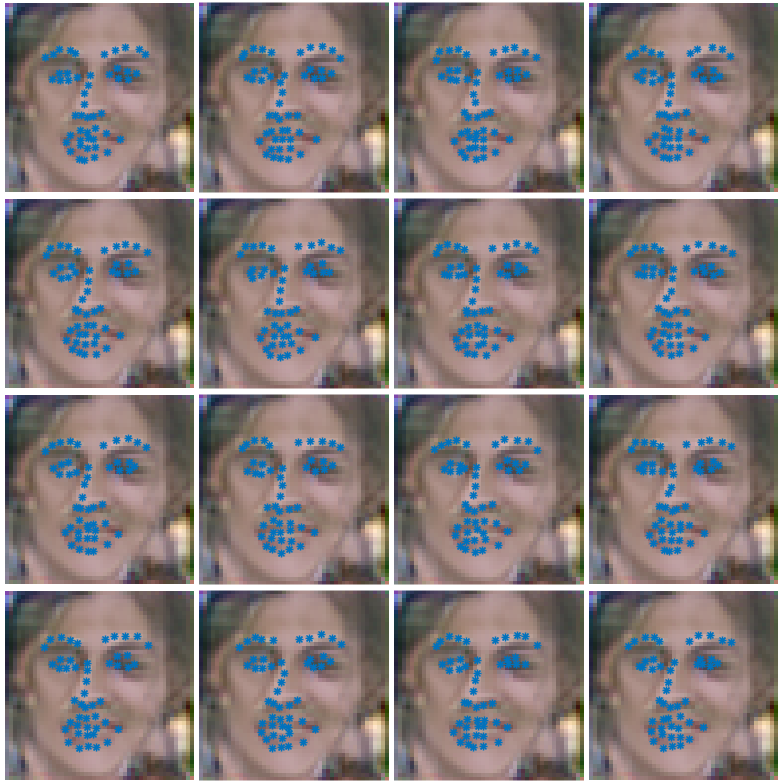}
\caption{cBiGAN.}
\label{fig:cBiGAN7911} 
\end{subfigure}
\caption{Generator outputs for proposed method (VAR+GAN) vs cBiGAN given particular landmarks.}
\label{fig:results7911}
\end{figure}

\begin{figure}
\begin{subfigure}[b]{.49\textwidth}
  \includegraphics[width=\columnwidth]{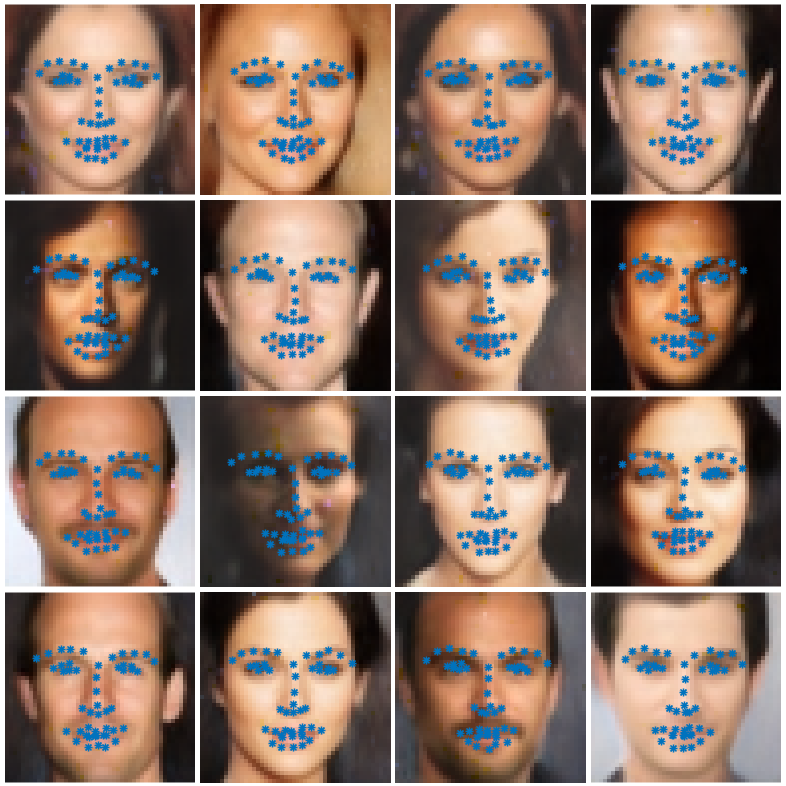}
\caption{Proposed method (VAR+GAN).}
\label{fig:VARGAN77911}       
\end{subfigure}
\begin{subfigure}[b]{.49\textwidth}
  \includegraphics[width=\columnwidth]{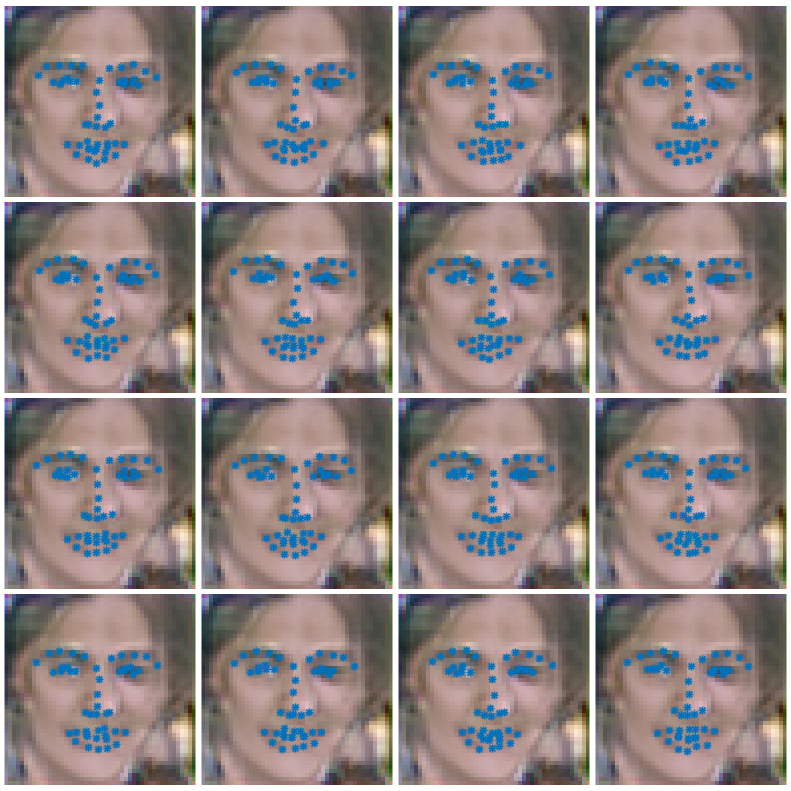}
\caption{cBiGAN.}
\label{fig:cBiGAN77911} 
\end{subfigure}
\caption{Generator outputs for proposed method (VAR+GAN) vs cBiGAN given particular landmarks.}
\label{fig:results77911}
\end{figure}
As shown in these figures, both methods are able to generate samples constrained to a particular set of landmarks but the proposed method generates higher variations of faces for a given landmark set while cBiGAN fails to create different samples in the same condition. The advantage of VAR+GAN is the versatility of the method which facilitates the implementation and also guarantees the higher quality in generated samples. For example in this work the proposed method is taking advantage of simplicity and power of BEGAN implementation and only change applied is to place a regression network and add its error value to the generator's loss. While cBiGAN method is constrained to a specific loss function which is a big disadvantage. 
\section{Conclusion and Future Work} 
In this work, a new scheme for training conditional deep generators has been introduced wherein the generator is able to constrain the generated samples to any continuous aspect of the dataset. The presented method is versatile enough to be applicable to any GAN variation with any network structure and loss function. The idea is to place a regression network in parallel with the discriminator network and back-propagate the regression error through the generator. A new loss function is also presented and it has been shown that it increases the JSD between data generated for any two set of aspects. The other property of the proposed loss is the reduction of the entropy for generated samples which is expectable because of the constraints applied to generated data.\\
The proposed method is also compared with the only available method with the same purpose (to the best of our knowledge in the time writing the article). cBiGAN method generates samples with a particular aspect but there are very low variations between generated samples while VAR+GAN produces higher variations for a specific set of aspects. Being able to generate variable samples is crucial for tasks including augmentation purposes. Being able to augment the database for certain data aspects and use them in training the final products is one of the most interesting applications for the conditional generators.\\
The future works include merging the VAC+GAN and VAR+GAN methods to constrain the generator to create samples from a specific class with a particular continuous aspect. And also investigating the influence of the generated samples in augmentation task for different applications.


\section*{References}
\bibliographystyle{elsarticle-num} 
\bibliography{lib}





\end{document}